\documentclass{mbe}
\newtheorem{theorem}{Theorem}

\newtheorem{lemma}{Lemma}

\theoremstyle{definition}
\newtheorem*{define}{Definition}
\usepackage{amssymb}

\begin{document}
\title{Is the Random Tree Puzzle process the same as the Yule--Harding process?}

\author{{%
Sha Zhu,$^\ast$ 
and 
Mike Steel$^\ast$}\\
\affiliation{$^\ast$Department of Mathematics and Statistics, University of Canterbury, 
Private Bag 4800, 
Christchurch, 8140, 
New Zealand}}

\authorlist{Zhu et al.}
\shorttitle{The Random Tree Puzzle process and the Yule--Harding process}

\mbeabstract{%
It has been suggested that a Random Tree Puzzle (RTP) process leads to a Yule--Harding (YH) distribution, when the number of taxa becomes large. In this study, we formalize this conjecture, and we prove that the two tree distributions converge for two particular properties, which suggests that the conjecture may be true. However, we present evidence that, while the two distributions are close, the RTP appears to converge on a different distribution than does the YH. }


\email{mike.steel@canterbury.ac.nz}


\section{Introduction}
The Maximum likelihood (ML) approach (Felsenstein 1981; Guindon and Gascuel 2003; Guindon et~al. 2010) is generally considered to be a reliable way of estimating phylogenies from DNA sequences.
However, ML is not always feasible for large numbers of species, because of the intensive computation required. Methods that use `four point subsets' (Dress et~al. 1986) reduce the complexity of the problem, and have assisted numerous studies. (Daubin and Ochman 2004; Nieselt-Struwe and {von Haeseler} 2001; Strimmer et~al. 1997; Strimmer and {von
  Haeseler} 1996).

The four points subtree is known as the quartet tree. {\it Quartet puzzling} (QP) (Strimmer and {von Haeseler} 1996) is an algorithm to infer a tree on $n$ taxa by using the quartet trees derived from DNA sequences. It firstly computes the likelihood of all $\binom{n}{4}$ quartets. As there are three possible topologies for any four taxa, the quartet tree which returns the greatest ML value is used (any ties are broken uniformly at random). At the puzzling step, the order of inserting new leaf nodes is randomized. A seed tree is built from the first four elements of the ordered leaf node sequence.  From this point on, leaves are attached sequentially by the following procedure:
when a new leaf $x$ is to be attached to the existing tree $T$,
quartet trees are built from quartets formed from $x$ and all subsets of size three are chosen from the existing leaf set. If the ML quartet tree of $\{i,j,k,x\}$ is $ij|kx$, then weight 1 is added to the edges on the path in $T$ connecting the two leaves $i$ and $j$. This process is repeated for all such quartet trees, and  $x$ is then attached to the edge which has the minimal weight. An example is given in Figure~\ref{fig:tp_alg}.

\begin{figure}[h]
\centering
\includegraphics[width=8cm]{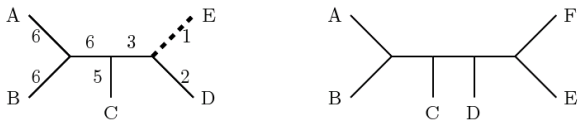}
\caption{Suppose leaf $F$ is about to be attached to the five-taxon tree on the left, and the ML trees of $\{i,j,k,F\}$ are $AB|CF$, $AC|EF$, $BC|DF$, $AC|DF$, $AB|DF$, $AD|EF$, $AB|EF$, $BC|EF$, $BD|EF$, and $CE|DF$. The external edge leading to $E$ returns the minimal weight, so $F$ is attached to this edge, leading to the six-taxon tree shown shown on the right.}\label{fig:tp_alg}
\end{figure}

Since the order of adding leaves is randomized, this can lead to variation in the resulting tree topologies, and so a consensus tree of numerous replicates is used as the output tree. The program {\it Tree-puzzle} (TP) (Schmidt et~al. 2002) is a parallel version of QP, which performs independent puzzling steps simultaneously.

The trees generated by either the QP or TP process depend on the biological sequences we have for the taxa. To investigate how the TP process behaves on randomized quartets, Vinh et~al. (2011) performed a simulation study on a so-called {\it random tree puzzle} (RTP) process. This assumes that no prior molecular information is given.
Therefore, for the same quartet set, all three tree topologies are equally likely.
The authors compare the empirical probabilities of tree topologies against the theoretical probabilities from the {\it proportional to distinguishable arrangement} (PDA) model and the {\it Yule-Harding} (YH) model.
Table 1 from Vinh et~al. (2011) reveals that the RTP's empirical probabilities are very close to the YH theoretical probabilities (indeed, there are two cases where these probabilities are identical). As it seems that the differences between the empirical and theoretical probabilities decrease as the number of taxa increases, Vinh et~al. (2011) suggest that the RTP process converges to the YH process as $n$ (the number of taxa) grows.
The authors provided further evidence for their conjecture by comparing some properties of  RPT trees with YH trees. Recall that a {\em cherry} in a tree is a pair of leaves that
are adjacent to the same vertex.  Then Vinh et~al. (2011) found that the mean and variance of the number of cherries were similar under the RTP simulation and the theoretical value under the YH process (McKenzie and Steel 2000).

Although Vinh et~al. (2011) provided evidence to suggest the two distributions appear to become very similar as $n$ grows, they did not provide a formal statement or proof of their claim that the two distributions converge. In this project, we investigate the RTP process further using mathematical and statistical methods.  
Our results demonstrate that certain properties of the trees that are near the `periphery' of the tree (i.e. near the leaves) converge under the two distributions;  however the `deep'
structure of the trees (how the tree is broken up around its centroid) appears to retain a trace that distinguishes the two models as the trees become large.

\section{Formalized Conjecture}
Given two discrete probability distributions $p$ and $q$ on $Y$, the {\it total variational distance} between $p$ and $q$ is defined as: $$d_{\rm VAR}(p,q)=\max_{A\subseteq Y}\left|\mathbb{P}_p(A)-\mathbb{P}_q(A)\right|, $$
where $\displaystyle{\mathbb{P}_p(A)=\sum_{y\in A}p(y)}$ and $\displaystyle \mathbb{P}_q(A)=\sum_{y\in A}q(y)$ are the probabilities of event $A$ under the distributions $p$ and $q$ respectively. Thus $d_{\rm VAR}(p,q)$ is the largest possible probability difference of any event under the distributions $p$ and $q$. A well-known and elementary result is that $\displaystyle d_{\rm VAR}(p,q)=\frac{1}{2}\sum_{y\in Y}\left|p(y)-q(y)\right| $, and thus the two distribution are the same if $d_{\rm VAR}(p,q)=0$.

A tree with the leaf set $X_n=\{1,2,\dots,n\}$ is called an {\it $X_n$-tree}. In the rest of this article, all $X_n$-trees referred to are binary trees, where the interior nodes have degrees of three. We use $T_n$ to denote a labeled $X_n$-tree topology, and $t_n$ to denote an unlabeled $X_n$-tree shape.
Vinh et~al. (2011) suggest that when the number of taxa ($n$)  becomes large, RTP converges to the YH distribution.
In this study, we consider the total variational distance between the tree topologies distributions between the RTP and the YH process, and formalize the  conjecture from Vinh et~al. (2011). This formalization states that the variational distance between the two tree distributions converges to zero as the number of taxa added  grows. We first note that it makes no difference to the truth of this conjecture whether the trees are labeled or unlabeled.

\begin{lemma}
Let $\mathcal{T}(n)$ and $\mathcal{S}(n)$ be the set of labeled and unlabeled $X_n$-trees respectively. For $T_n \in \mathcal{T}(n)$, and $t_n \in \mathcal{S}(n)$, let $\displaystyle\Delta_n:=\sum_{ T_n \in \mathcal{T}(n) } \left|\mathbb{P}_{\rm YH}(T_n)-\mathbb{P}_{\rm RTP}(T_n) \right|$ and $\displaystyle\delta_n:=\sum_{ t_n \in \mathcal{S}(n) } \left|\mathbb{P}_{\rm YH}(t_n)-\mathbb{P}_{\rm RTP}(t_n) \right|$. Then, $\Delta_n=\delta_n,$
and in particular $\displaystyle\lim_{n\rightarrow\infty}\Delta_n=0 \Longleftrightarrow  \lim_{n\rightarrow\infty}\delta_n=0,$
as $n\rightarrow\infty$.
\end{lemma}
\begin{proof}
Let $\nu(t_n)$ be the number of $X_n$-trees $T_n$ that have the shape $t_n$. Then, for $*\in\{{\rm{YH}},{\rm{RTP}}\}$, $\displaystyle\mathbb{P}_*(T_n)=\frac{\mathbb{P}_*(t_n)}{\nu(t_n)}$, we have:
\begin{align*}
\Delta_n&=\sum_{ T_n \in \mathcal{T}(n) } \left|\mathbb{P}_{\rm YH}(T_n)-\mathbb{P}_{\rm RTP}(T_n) \right|\\
&=\sum_{ t_n \in \mathcal{S}(n) } \sum_{ \stackrel{T_n\in \mathcal{T}(n)}{T_n \text{ has shape }t_n}}\left|\mathbb{P}_{\rm YH}(T_n)-\mathbb{P}_{\rm RTP}(T_n) \right|\\
&=\sum_{ t_n \in \mathcal{S}(n) } \nu(t_n) \left|\frac{\mathbb{P}_{\rm YH}(T_n)}{\nu(t_n)}-\frac{\mathbb{P}_{\rm RTP}(T_n)}{\nu(t_n)} \right|\\
&=\sum_{ t_n \in \mathcal{S}(n) } \left|\mathbb{P}_{\rm YH}(t_n)-\mathbb{P}_{\rm RTP}(t_n) \right|\\
&=\delta_n.
\end{align*}

\end{proof}

Thus, we formalize the conjecture from Vinh et~al. (2011) as follows:\\

{\noindent{\bf Conjecture (strong version)}

With $\Delta_n=\delta_n$ defined as above, $\displaystyle \lim_{n\rightarrow\infty}\Delta_n=0.$
}

\bigskip

Note that, in the YH process, new leaves are only ever attached to pendant edges, and each pendant edge is selected with equal probability. We say that such leaves are attached to {\em uniformly selected pendant edges}.
 By contrast, the RTP process can attach new leaves to any edge, although RTP has an increasingly strong preference to attach leaves to pendant edges as the tree grows (Vinh et~al. 2011). These authors also suggested that as the tree grows, the number of cherries of a RTP tree follows the same limiting distribution as the number of cherries of a YH tree, which is normally distributed.
We summarize these two claims as follows:\\
{\noindent{\bf Conjecture (weak version)}

\begin{enumerate}
\item\label{conj_p}Let $\mathcal{E}_m$ be the event that {\em all} leaf attachments under the RTP beyond the first $m$ leaves,  are to uniformly selected pendant edges.
Then
$\mathbb{P}(\mathcal{E}_m)\rightarrow 1$, as $m$ tends to infinity.
\item\label{conj_cherry} The distribution of cherries converges to the same (asymptotic) normal distribution as the YH model.
\end{enumerate}
}

In our paper, we prove the two parts of the weak conjecture, and present statistical evidence that the strong conjecture is not true.

\section{RTP is similar to YH when {\it n} is large}
To verify Part \ref{conj_p} of the weak conjecture, we need to establish that the probability that a new leaf attaches to a pendant edge converges to 1 sufficiently quickly as the number of leaves increases. This requires that the pendant edges carry less weight than the interior edges. In addition, when the new leaf is added, all pendant edges must be equally likely to be chosen. Thus we must check the edge weight distribution during the puzzling step of the RTP process.

\subsection{Distribution of edge weights}

Let $E_n^{\rm P}$ denote the set of pendant edges of current $X_n$-tree $T_n$ and let $E_n^{\rm I}$ be the set of interior edges. For any edge $e$ of $T_n$, we let $W(e)$ denote the random variable edge weight during the quartet puzzling step.
Suppose edge $e$ has $k$ leaves of $T_n$ on one side and $n-k$ leaves of $T_n$ on the other side.  The following result is established in the Appendix.

\begin{lemma}
\label{centrallem}
$W(e)$ is a binomial random variable with the parameters $\frac{k(n-k)(n-2)}{2}$ as the number of trials and $\frac{2}{3}$ as the probability of success on each trial.
\end{lemma}

The parameter $k$ takes the value $1$ or $n-1$ for a pendant edge;  for an interior edge, $k$ lies between $2$ an $n-2$. Next, we show that for any fixed pendant and interior edge, the probability that the interior edge has lower weight converges to zero exponentially fast with increasing $n$.
More precisely, for any $e''\in E_n^{\rm P}$ and any $e'\in E_n^{\rm I}$, we establish the following result in the Appendix.
\begin{equation}
\mathbb{P}\left(W_n(e'')\geq W_n(e')\right)\leq 2\exp(-\frac{1}{576}n)\label{ineq:weight}.
\end{equation}
This result is for a fixed pair of pendant and interior edges, but it easily implies  that the probability that the smallest weight in the tree is on a pendant rather than an interior edge converges quickly to 1 with increasing $n$. This is formalized in the following inequality, also proved in the Appendix:
\begin{equation}
\mathbb{P}\left(\min_{e\in E_n^{\rm P}}\{W_n(e'')\} \leq \min_{e'\in E_n^{\rm I}}\{W_n(e')\}\right)\geq 1- 2 n^2\exp(-\frac{1}{576}n)\label{ineq:min_weight}.
\end{equation}
Thus a new leaf is almost certain to be added to pendant edges; moreover, as noted above, each pendant edge has equal probability of being attached to.

\subsection{New leaves attach rarely to interior edges}

\begin{theorem}\label{thm:more_than_m}
Suppose $T_m\in \mathcal{T}(m)$, let $\mathcal{E}_m$ be the event that all leaf attachments under RTP beyond $T_m$ are to uniformly selected pendant edges. Then, for constants $a, b>0$:
$$\mathbb{P}(\mathcal{E}_m)\geq 1- ae^{-bm}.$$ 
\end{theorem}

\begin{proof}
Let $B_k$ be the event that $(k+1)-$st leaf  is not attached to any leaf edge of $T_k$. Then we have $1-\mathbb{P}(\mathcal{E}_m)=\mathbb{P}\left(\bigcup^\infty_{k=m} B_k\right)$. By Boole's inequality, we have $\mathbb{P}\left(\bigcup^\infty_{k=m} B_k\right)\leq \sum^\infty_{k=m} \mathbb{P}\left( B_k\right)$.
By Inequality \eqref{ineq:min_weight},
$\mathbb{P}\left( B_k\right)\leq 2k^2\exp(-\frac{1}{576}k)$.
We now use the following general inequality, the proof of which is given in the Appendix. If $\displaystyle Q_m=\sum^{\infty}_{k=m} k^2\exp(-c k)$, where $c\geq\frac{4\log{k}}{k}$
and $k> 1$, then for $m\geq m_0$: 
\begin{equation}
\displaystyle Q_m\leq \frac{\exp(-cm_0/2)}{1-\exp(-c/2)}.\label{ieq:m_is_big}
\end{equation}
Thus,
\begin{align*}
1-\mathbb{P}(\mathcal{E}_m) &\leq \sum^\infty_{k=m}2 k^2\exp(-\frac{1}{576}k) \\ &\leq \frac{2}{1-\exp(-\frac{1}{576}\times\frac{1}{2})}\exp(-\frac{1}{576}\times\frac{1}{2}m). 
\end{align*}
Rearranging this inequality establishes the inequality in the  theorem.  The uniformity follows by Lemma~\ref{centrallem}.
\end{proof}

\subsection{The mean and variance of the number of cherries in the RTP tree}
Table 3 of Vinh et~al. (2011) reveals that the mean and variance of the number of cherries on trees generated under the RTP process and under YH process are similar. In order to provide a formal proof that they converge to the same limiting distribution, we need to introduce the {\it Extended Poly\'{a} urn model} (EPU).
\subsubsection{Extended Poly\'{a} urn model}
Consider the following extended  Poly\'{a} urn (EPU) model: at time $t=0$, there are $b$ blue balls and $r$ red balls in an urn, where $b\geq 0$ and $r\geq 0$. At each discrete time step, one ball is picked at random from the urn. If the ball is blue,  $c$  additional blue balls and $d$ red balls will be placed; if the picked ball is red, $e$ additional  blue balls and $f$ red balls will be placed. The values $c,d,e,f$ can also take negative values, in which case, instead of placing new balls in the urn, the number of balls of the appropriate colour will be withdrawn. We use $b_n$ to denote the number of blue balls after the $n$th draw, and $S_n$ is the total number of balls. The following matrix describes this process:
$$A=\left[\begin{array}{cc}
   c&d\\e&f\\
  \end{array}\right].
$$
We require that $A$ has positive and equal row sums, as well as one real positive principal eigenvalue $\lambda$.  Let $\left[\begin{array}{c}v_1\\v_2 \end{array}\right]$ be the normalized eigenvector associated with $\lambda$. Then, under these conditions, a classic result 
states that, as $n\rightarrow\infty$, 
$\frac{b_n-\lambda v_1 n}{\sqrt{n}}\stackrel{\mathcal{D}}{\rightarrow} \mathcal{N}(0,\sigma^2)$ (Mahmoud 2008; Bagchi and Pal 1985),
where $\stackrel{\mathcal{D}}{\rightarrow}$ denotes convergence in distribution. 
Crucially, the initial values of $b$ and $r$ do not play any significant roles in this limiting normal distribution (or of its mean and variance).  

\subsubsection{EPU and attaching new edges only to pendant edges}
We relate the Yule process to the EPU model as follows: consider the set of cherry edges as a collection of blue balls, and the non-cherry edges as a collection of red balls. When a new edge is attached to a pendant edge, if it is attached to a cherry edge, the number of cherry edges remain the same, but the number of non-cherry edges increases by one. If a new edge is added to an non-cherry edge, then the non-cherry edge becomes a cherry edge, and the new edge is also a cherry edge. Thus, the generating matrix is:
$$A=\left[\begin{array}{cc}
   0&1\\2&-1\\
  \end{array}\right].
$$
Notice that $A$ has row sum equal to $1$ and $A$ has one real positive eigenvalue $\lambda$, as required.

Let $C_n$ be the number of cherries in a YH tree. Then as $n$ tends  to infinity,  $$Z_n := (C_n-n/3)/\sqrt{2n/45}$$ converges in
distribution to a standard normal distribution (i.e.  $Z_n \xrightarrow{\mathcal{D}}  N(0,1)$), by  Corollary 3 of (McKenzie and Steel 2000).
We now show that the same holds for the distribution of cherries in an RTP tree. 

\begin{theorem}\label{thm:z_star_converge2_z}
Let $C^*_n$ be the number of cherries in an RTP tree, and let  $Z^*_n = (C^*_n-n/3)/\sqrt{2n/45}$. 
Then $Z^*_n \xrightarrow{\mathcal{D}}  N(0,1)$.
\end{theorem}

\begin{proof}
We need to show that for any $\epsilon>0$, and for all sufficiently large value of $n$ and all positive real $x$,
\begin{equation}
 |\mathbb{P}(Z^*_n<x ) - \mathbb{P}(Z<x)|\leq \epsilon.\label{eqn:z_star_converge2_z}
\end{equation}
where $Z$ is a standard normal random variable.

As before, let $\mathcal{E}_m$ be the event that after $m$ leaves have been attached to the starting tree by RTP, all further additions are to pendant edges, and let $\mathcal{E}_m^c$ be the complement of $\mathcal{E}_m$. For $n>m$, by the law of total probability, we have:
\begin{equation}
\mathbb{P}(Z^*_n<  x) = \mathbb{P}(Z_n^*<x|\mathcal{E}_m)\mathbb{P}(\mathcal{E}_m) + \mathbb{P}(Z_n^*<x|\mathcal{E}_m^c)\mathbb{P}(\mathcal{E}_m^c).\label{eqn:z_star}
\end{equation}
If we now subtract $\mathbb{P}(Z^*_n<x|\mathcal{E}_m)$ from both side of Equation \eqref{eqn:z_star}, we obtain:
\begin{equation}
\begin{array}{rl}
&\mathbb{P}(Z^*_n<x) - \mathbb{P}(Z^*_n<x|\mathcal{E}_m)\\
=&\mathbb{P}(Z_n^*<x|\mathcal{E}_m)(\mathbb{P}(\mathcal{E}_m)-1) + \mathbb{P}(Z_n^*<x|\mathcal{E}_m^c)\mathbb{P}(\mathcal{E}_m^c).
\end{array}\label{eq:step2}
\end{equation}
By the triangle inequality ($|a+b|\leq |a|+|b|$) we have:
\begin{equation}
\begin{array}{rl}
&|\mathbb{P}(Z_n^*<x|\mathcal{E}_m)(\mathbb{P}(\mathcal{E}_m)-1)+ \mathbb{P}(Z_n^*<x|\mathcal{E}_m^c)\mathbb{P}(\mathcal{E}_m^c)| \\
\leq  & |\mathbb{P}(Z_n^*<x|\mathcal{E}_m)(\mathbb{P}(\mathcal{E}_m)-1)|+ | \mathbb{P}(Z_n^*<x|\mathcal{E}_m^c)\mathbb{P}(\mathcal{E}_m^c)|.
\end{array}
\label{ieq:tri_ineq}
\end{equation}
Combining Equation \eqref{eq:step2} and Inequality \eqref{ieq:tri_ineq} gives the following:
\begin{equation}
\begin{array}{rl}
&|\mathbb{P}(Z^*_n<x) - \mathbb{P}(Z^*_n<x|\mathcal{E}_m)|\\
\leq & |\mathbb{P}(Z_n^*<x|\mathcal{E}_m)(\mathbb{P}(\mathcal{E}_m)-1)|+ |\mathbb{P}(Z_n^*<x|\mathcal{E}_m^c)\mathbb{P}(\mathcal{E}_m^c)| ,\\
\leq &|\mathbb{P}(Z_n^*<x|\mathcal{E}_m)||(\mathbb{P}(\mathcal{E}_m)-1)|+ |\mathbb{P}(Z_n^*<x|\mathcal{E}_m^c)||\mathbb{P}(\mathcal{E}_m^c)| .
\end{array}\label{ieq:step2}
\end{equation}

Theorem \ref{thm:more_than_m} tells us that $\mathbb{P}(\mathcal{E}_m)\geq  1-ae^{-bm}$, which  tends to $1$ as $m$ grows.
Now, since $\mathbb{P}(\mathcal{E}_m^c)\rightarrow0$ as $m$ tends to infinity, we can select a sufficiently large value of $m$ that $\mathbb{P}(\mathcal{E}_m^c)\leq \epsilon/4$ and $\mathbb{P}(\mathcal{E}_m)\geq  1-\epsilon/4$. Thus, $\mathbb{P}(\mathcal{E}_m)-1 \geq  -\epsilon/4$, and $|\mathbb{P}(\mathcal{E}_m)-1 |\leq  \epsilon/4$. Since $0\leq \mathbb{P}(Z_n^*<x|\mathcal{E}_m)$, $\mathbb{P}(Z_n^*<x|\mathcal{E}_m^c)\leq 1$, Inequality \eqref{ieq:step2} gives:
\begin{equation}
 |\mathbb{P}(Z^*_n<x) - \mathbb{P}(Z^*_n<x|\mathcal{E}_m)|\leq \epsilon/4+\epsilon/4=\epsilon/2 \label{eqn:z_star_n},
\end{equation}
for all sufficiently large $m$, and all $n \geq m$ and $x>0$.

Now we consider the sequence of $Z^*_n$ conditional on $\mathcal{E}_m$. By conditioning on this event all the new leaves are to uniformly selected pendant edges. 
Because the EPU argument that established the convergence of the sequence $Z_n$ (the normalization of the number of cherries in a YH tree) does not depend on the initial number of cherries for any $\epsilon>0$, and every $m$,  there exists an integer $n_0$ so that for all $n \geq n_0$, and $x>0$: 
\begin{equation}
|\mathbb{P}(Z^*_n<x|\mathcal{E}_m) - \mathbb{P}(Z_n<x)|\leq \epsilon/2. \label{eqn:z_star_z_n}
\end{equation}
Then, by the triangle inequality ($|a+b|\leq |a|+|b|$), if we  add Inequalities \eqref{eqn:z_star_n} and \eqref{eqn:z_star_z_n}, we have
$$|\mathbb{P}(Z^*_n<x) - \mathbb{P}(Z_n<x)|\leq \epsilon,
$$
and since $Z_n$ converges in distribution to a standard normal, this establishes \eqref{eqn:z_star_converge2_z}.

\end{proof}
Theorem \ref{thm:z_star_converge2_z} shows that the number of cherries on the RTP trees has a limiting normal distribution with the same asymptotic mean and variance as for the YH distribution.

We have also shown that, from some point forward, new leaves will always be added to pendant edges, which verifies the weak conjecture.
While these two results may be regarded as providing some weak evidence in favour of the strong conjecture, they do not constitute any formal justification of it.
In the next section, we will provide an analysis that suggests that the variational distance between the two distributions remains bounded away from zero as $n$ grows, and this makes these two process distinct in the limit.

\section{Is RTP the same as YH?}

Consider the following scenario where we perform the YH process on some starting tree with more than three leaves, where $v$ is one of the interior nodes. At node $v$, the graph is divided into three subtrees (see Fig.~\ref{fig:centroid}). We let $L_i$, $(i=1,2,3)$ denote the leaf sets of these subtrees, and let $l_i=|L_i|$, $(i=1,2,3)$ denote the number of leaves in the sets. We normalize the $l_i$ values by the total number of leaves $n$. Clearly, the sequence of $l_i/n$ values change, as new leaves are gradually added to the whole tree.

\subsection{Poly\'{a} urns and the centroid of a tree}
Adding new leaves on to the tree under the YH process ensures that each new leaf is always added into one of the leaf sets $L_i$, $(i=1,2,3)$. The probability that $l_i$ increases by one is the relative proportion of the number of leaves of the subtree in relation to the number of leaves in the full tree. This is similar to the Poly\'{a} urn problem (Karr 1993) involving balls of three different colours.

Suppose that one ball is picked randomly at each step, and replaced along with another ball of the same colour into the urn.
Let $F_n^i$ be the relative frequency of the $i$th colour ball when $n$ balls are present, and $\mathbf{F}_n=(F_n^1,F_n^2,F_n^3)$. Then $\mathbf{F}_n$ converges (as $n \rightarrow \infty$) to a  Dirichlet distribution (Kotz et~al. 2000) with the parameter vector $\mathbf{F}_{n_0}$, where $n_0$ is the total initial number of balls.  Different initial values in the urn produce different distributions when $n$ balls are present in the urn, and this difference in distributions does not converge to zero as $n$  grows.  This result suggests that the YH process on different initial $X$-trees may well lead to different distributions of the resulting trees. However, if the final tree shape is the only information we are given, then it will be impossible to identify the position of the original vertex $v$ in the final tree with certainty. Thus the frequencies $\mathbf{F}_n$ cannot be clearly measured from the final tree alone. However, we can partly ameliorate this problem by considering a particular vertex that we can 
easily identify in the final tree, namely its centroid (Jordan 1869; Mitchell 1978).

\begin{figure}[ht]
\centering
\includegraphics[width=5cm]{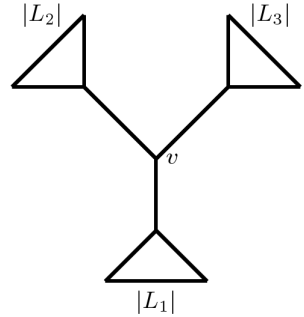}
\caption{Centroid of a tree}\label{fig:centroid}
\end{figure}

\begin{define}
A vertex $v$ of a tree $T=(V,E)$ is a {\it centroid} if each component of the disconnected graph $T\backslash v$ has, at most $(1/2)|V|$ vertices.
\end{define}
A well known property of centroids states that a tree has either a single centroid or two adjacent centroids, in which case $|V|$ is even (Kang and Ault 1975). To keep the problem simple, we only consider trees with a single centroid. However, because $T$ is a binary tree, $|V|$ is always even, and so this does not guarantee a unique centroid. Fortunately, the following lemma shows that a binary tree with odd number of leaves always has a unique centroid.

\begin{lemma}
Let $T$ be an unrooted binary $X_n$-tree. Then:
\begin{enumerate}
\item A vertex $v$ of $T$ is a centroid of $T$ if and only if $v$ satisfies $l_1,l_2,l_3\leq \frac{n}{2}$, where $l_i$ are the number of leaves of the three subtrees of $T\backslash v$.
\item If $n$ is odd, then $T$ has a unique centroid.
\end{enumerate}
\end{lemma}

\begin{proof}
\begin{enumerate}
\item[(1)]  Suppose that $v$ is an interior vertex of $T$. Consider the vertex sets $V_1$, $V_2$ and $V_3$ of the connected components of $T \backslash v$. Let $l_i$ be the number of leaves in $V_i$. Considering the rooted binary tree on $V_i$, we have:
\begin{equation}
|V_i|=2l_i-1 .\label{eqn:v_eq_2l_minus_1}
\end{equation}
Also, since $T$ is an unrooted binary tree, we have:
\begin{equation}
|V|=2n-2 . \label{eqn:v_eq_2n_minus_2}
\end{equation}

Thus, $|V_i| \leq \frac{1}{2}|V|$ if and only if $2l_i-1 \leq \frac{1}{2}(2n-2)$ and this holds precisely if $l_i \leq n/2$.  Thus, the condition for $v$ to be a centroid (namely that
$|V_i| \leq \frac{1}{2}|V|$ for $i=1,2,3$) is precisely the same as that stated in the lemma.

\item[(2)]  Suppose $v$ is a centroid of $T$. At $v$, we let $L_i$, ($i=1,2,3$) denote the leaf set of the subtrees $T_i$ and let $l_i$ denote the size of these leaf sets, ordered so that $l_j\leq l_3\leq \frac{|X|}{2}$, ($j=1,2$).  Since $n$ is odd, we have $l_3 < \frac{n}{2}$.

Suppose another centroid $d$ exists. We use $L_i'$ to denote the complement of $L_i$. Then there is a subtree $H$ of $T$ rooted at $d$, with leaf set $L_H$, where $L_H \supseteq G'$, and $G'\in \{L_1',L_2',L_3'\}$. Since $l_j\leq l_3< \frac{n}{2}$, where $j\in \{1,2\}$, we then have $|L_H|\geq |G'|> \frac{n}{2}$. Therefore, $d$ cannot be a centroid.

\end{enumerate}

\end{proof} 
\begin{figure}[ht]
\centering
\includegraphics[width=8cm]{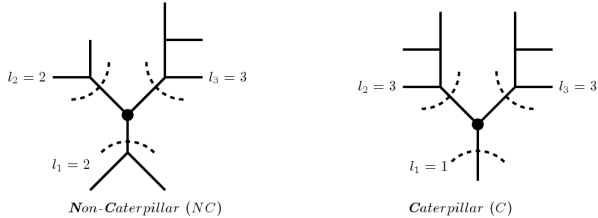}
\caption{The two tree shapes for binary trees on seven leaves}\label{fig:7tax}
\end{figure}

We now relate the centroid back to the Poly\'{a} urn problem. First notice that tree shapes only start to differentiate when there are more than five leaves. Therefore, in the following scenario, we perform the YH process from initial trees with seven leaves. Suppose that a tree $X$ is either the non-caterpillar (NC) or caterpillar (C) tree shown in Fig.~\ref{fig:7tax}.
We will use  $X$ as the initial tree to construct some tree $t_n$. At the centroid of $t_n$ when $n=7$ the sequences of  $l_i/n$ are $(2/7,2/7,3/7)$ and $(1/7,3/7,3/7)$ for $t_7=NC$ and $t_7=C$ respectively. Now, let us only consider the number of leaves $l_1$ in the smallest subtree of $t_n$ for all odd values of $n \geq 7$ (henceforth all values of $n$ in this section are odd to guarantee a unique centroids, and limits as $n$ tends to infinity are also over just the odd values of $n$). We define the ratio of $l_1$ and of number of leaves $n$ as $\pi_n^X =\frac{l_1}{n}$.  For $\gamma \in (0,1)$, let $\Pi^X$ be the limiting probability of the event $ \pi_n^X \geq\gamma$. In other words, $\displaystyle \Pi^X=\lim_{n\rightarrow \infty} \mathbb{P}(\pi_n^X \geq \gamma)$. To test the null hypothesis that $ \Pi^{NC} =\Pi^C$,
we investigate the ratio $\pi_n^X$ under the YH process. An additional 2000 leaves are attached to the starting trees $NC$ and $C$ under the YH process with 1000 replicates each case.
Using the initial tree $NC$ or $C$, we found that the probability that $\pi_n^X$ is greater than $\gamma=0.19$ does not appear to be converging for the two choices of X ($NC$ or $C$) (see Fig.~\ref{fig:7yuleprob}).
Fig.~\ref{fig:7yuleprob} indicates the $95\%$ confidence interval of proportions of the event for which  $\pi_n^X \geq  0.19$, which suggests the following strict inequality:
\begin{equation}
 \Pi^{NC}>\Pi^C.\label{ineq:piAB}
\end{equation}
\begin{figure}[ht]
\begin{center}
\includegraphics[width=8cm]{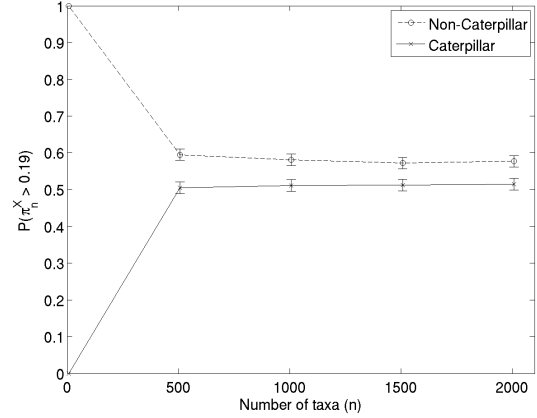}
\caption{Empirical probabilities and the $95\%$ confidence interval proportion of the event that $\pi_n^X \geq 0.19$. The dashed line is for the initial tree of the non-caterpillar seven-taxa tree; and solid line is for the caterpillar seven-taxa tree.}
\label{fig:7yuleprob}
\end{center}
\end{figure}

\subsection{A modified RTP process}
To provide evidence that the RTP and the YH processes are not exactly the same, we define a new process {\bf RTP$'$},  which is equivalent to the RTP process up to $n=7$.  From this point  forward it proceeds according to the YH process.
Therefore, the initial probabilities of constructing $X_n$-trees from $NC$ and $C$ under the RTP$'$ process are different from the YH process.
We use the probabilities of the starting tree $NC$ and $C$ under the RTP process as
the probabilities under the RTP$'$. Vinh et~al. (2011) estimated by simulations that the probabilities for the seven-taxa non-caterpillar tree is 0.4607 under the RTP process and 0.4667 under the YH process, which gives us the following inequality:
\begin{equation}
\mathbb{P}_{YH}(t_7=NC)-\mathbb{P}_{RTP'}(t_7=NC)>0\label{ineq:vinh}.
\end{equation}

\begin{theorem}
\label{thm33}
If \eqref{ineq:piAB} holds then $$\displaystyle\lim_{n\rightarrow\infty}d_{\rm VAR}(\mathbb{P}_{\rm RTP'}(t_n),\mathbb{P}_{\rm YH}(t_n))\neq0.$$
\end{theorem}
\begin{proof}
Let $\mathcal{S}(n)$ be the set of unlabeled $X_n$-tree and let:
\begin{equation}
\label{deltaeq}
\delta':=\sum_{t_n\in \mathcal{S}(n)}|\mathbb{P}_{\rm YH}(t_n)-\mathbb{P}_{\rm RTP'}(t_n)|.
\end{equation}
Consider the event $\Sigma_n$ that $\pi_n^X \geq \gamma$. Then:
\begin{eqnarray}
\mathbb{P}_{\rm YH}(\Sigma_n)=\sum_{X\in\{NC,C\}}\mathbb{P}_{\rm YH}(\Sigma_n|t_7=X) \mathbb{P}_{\rm YH}(t_7=X)\label{eqn:yh}\\
\mathbb{P}_{\rm RTP'}(\Sigma_n)=\sum_{X\in\{NC,C\}}\mathbb{P}_{\rm RTP'}(\Sigma_n|t_7=X)\mathbb{P}_{\rm RTP'}(t_7=X) \label{eqn:tpstar}
\end{eqnarray}
If we now subtract Eqns. \eqref{eqn:tpstar} from \eqref{eqn:yh}, and substitute $\mathbb{P}_*(t_7=C)$ in \\$1-\mathbb{P}_*(t_7=NC)$, we have:
\begin{equation}
\begin{array}{rl}
&\mathbb{P}_{\rm YH}(\Sigma_n)-\mathbb{P}_{\rm RTP'}(\Sigma_n)\\
=&\left(\mathbb{P}_{\rm YH}(t_7=NC)-\mathbb{P}_{\rm RTP'}(t_7=NC)\right)(\Pi^{NC}-\Pi^C).
\end{array}\label{eqn:diff_yh_tpstar}
\end{equation}
Thus, if we apply inequalities \eqref{ineq:vinh} and \eqref{ineq:piAB} in Eqn. \eqref{eqn:diff_yh_tpstar}, we obtain $\mathbb{P}_{\rm YH}(\Sigma_n)-\mathbb{P}_{\rm RTP'}(\Sigma_n)>0$.  Consequently,  $\delta'>0$ in (\ref{deltaeq}), and so 
$\displaystyle\lim_{n\rightarrow\infty}d_{\rm VAR}(\mathbb{P}_{\rm RTP'}(t_n),\mathbb{P}_{\rm YH}(t_n))\neq0$, as claimed.
\end{proof}

It is important to be clear about what we have established: we have not formally shown that RTP does not converge to YH, nor even that RTP$'$ fails to converge to YH. Rather, we have provided evidence that a certain property of RTP$'$ holds, and if so, this implies (Theorem~\ref{thm33})  that RTP$'$ does not converge to YH.
Then, since RTP$'$ is a hybrid of YH and RPT,  this suggests that RPT does not either.

\section{Further discussion and concluding comments}

In phylogenetic studies, trees are inferred from DNA sequences using various methods. It is also pertinent to ask what sort of trees these methods would produce, given entirely random data.
This is one of the motivations of the study by Vinh et~al. (2011). In the following discussion, we use an $n$ by $k$ matrix $D$ to denote a sequence of $k$ independent characters on $n$ taxa. Note that all the characters have the same state space $S$. The term `random data' can refer to any one of the following three schemes:
\begin{enumerate}
\item[{\bf (R1)}.] State $x$ is assigned to taxon $i$ in character $j$ by an independent, identically distributed (i.i.d.) process with a probability $p_j(x)$, for $x\in S$.
\end{enumerate}
When the probabilities of state $x$ are the same for all characters (i.e. if $p_j(x)=p(x)$ for all $j$),  we obtain a stronger notion as follows:
\begin{enumerate}
\item[{\bf (R2)}.] For every entry of the matrix $D$, $D_{ij}$ is assigned to state $x$ with probability $p(x)$.
\end{enumerate}
If all states are equally likely (i.e. if $p(x)=1/|S|$), we arrive at an even stronger notion as follows:
\begin{enumerate}
\item[{\bf (R3)}.] For all entries of $D$, all states have equal probabilities.
\end{enumerate}

Vinh et~al. (2011) suggest that random data imply that quartet trees are equally likely and independent to each other, stating:
\begin{quote}
{\it
In our setting, we assume no phylogenetic information
in the data. This is equivalent to the assumption that each
of the three topologies for a quartet is equally likely and
that the tree topology for each quartet is independent of
the other quartets. \dots ~Hence, $3^{\binom{n}{4}}$ possible combinations
of quartet trees will serve as input to TP. }
\end{quote}
For any of the models (R1)--(R3), it certainly is true that random sequence data provide equal support for all three possible topologies of any four taxa. However, this does not necessarily imply that the inferred quartet trees are exactly independent. Rather than persue this question here, we will consider the behavour of TP under a model in which quartet trees are i.i.d. and uniform, as in Vinh et~al. (2011).

While the RTP process appears to converge close to the YH distribution, it is instructive to note that another tree reconstruction method, {\it maximum parsimony} (MP), when applied on random data, converges to a quite different distribution on trees.  Under model (R3) with two states MP converges to  the PDA (`proportional to distinguishable arrangements') model, which selects each unrooted binary tree with equal probability.   Let $B(n)$ be the set of unrooted binary trees on the leaf set $\{1,2,\ldots n\}$. For model (R3) with two states and $k$ independent characters, we use $\mathcal{T}_{MP}(D)$ to denote the MP tree on $D$ (if the MP tree for $D$ is not unique then select one MP tree uniformly at random).

\begin{theorem}\label{thm:mp}
Under random model {\bf (R3)} with two states:
\begin{enumerate}
\item\label{thm:mp1} The random tree $\mathcal{T}_{MP}(D)$ has a PDA distribution on $B(n)$; i.e. $$\mathbb{P}(\mathcal{T}_{MP}(D)=T)=\frac{1}{|B(n)|}.$$
\item\label{thm:mp2} For each fixed $n$, there is a unique MP tree for $D$ with probability converging to 1 as $k$ grows.
\end{enumerate}
\end{theorem}

\begin{proof}
\begin{enumerate}
\item Let $w(D,T)$, $T\in B(n)$, denote the parsimony score of $T$ on random data $D$. By Theorem 7.1 of Steel (1993), the number of ways to colour the leaves
 of a binary tree $T$ with $n$ leaves with using two colours, and so that the resulting colouration has parsimony score of $k$ for $T$ depends only on $n$ and not otherwise on the tree $T$. Hence, for all $T\in B(n)$, the probability $\mathbb{P}(w(D,T)=l)=f(l)$, is the same for all binary trees with a given number of leaves.  Therefore, each tree has the same probability of being an MP tree for $D$.

Let $E_k(T, T')$ be the event that $T$ and $T'$ have exactly the same parsimony score. By the Central Limit Theorem, the probability that the difference in scores is exactly 0 (i.e. $\mathbb{P}(E_k(T,T'))$) tends to zero as $k$ grows.

Let $E$ be the event that the maximum parsimony tree for $D$ is unique, and let $E^c$ be the complement, namely that there are at least two trees which have the same parisimony score for $D$. Note that $E^c$ is a subset of the union of the events $E_k(T,T')$ over all $T,T'$ (distinct). Therefore, we have: $$1-\mathbb{P}(E) \leq \mathbb{P}(\displaystyle{\bigcup_{T,T'}} E_k(T,T')) \leq \sum_{T,T'} \mathbb{P}(E_k(T,T')) \rightarrow 0, $$ as $k$ grows. Thus, $\mathbb{P}(E)\rightarrow 1$, as $k\rightarrow \infty$, as required.
\end{enumerate}
\end{proof}
Hence the MP tree on random data with two states converges to the PDA model.

In the PDA model, new leaf nodes are uniformly added onto any edges of the existing tree, whereas the Yule tree selects a pendant edge randomly, and adds a new node onto this pendant edge.
During the construction process, PDA, RTP and RTP$'$ can attach some new leaves onto interior edges.  For the PDA process, this has probability of almost $1/2$, and it is much less for RTP, as the number of leaves increases. In the case of RTP$'$, beyond seven leaves, all further leaves are inserted to a pendant edge, just as in  the YH model.

In conclusion, we have verified that the RTP process will eventually not add new leaves onto interior edges after some point, which makes the RTP process become more like the YH process. However, the distance between two distributions appears to remain bounded away from zero even when $n$ tends to infinity, which suggests that they are still two distinct tree construction methods.

\section{Acknowledgments}
We thank Marsden Fund for supporting this work. We also thank
David Aldous for suggesting we consider the stochastic properties of RTP$'$ trees relative to their centroids.

\section*{Appendix: Technical details}

\subsection*{\bf Proof of Lemma 2}
\begin{proof}
At edge $e$, suppose that $A$ and $B$ partition $X_n$, where $n-1\geq k\geq 1$, $|A|=k$ and $|B|=n-k$.
Let $\{a,b,c\}$ be a subset of $X_n$ of size three. Suppose that a new leaf $x$ is to be attached to $e$. Let $q$ be a split of $\{x,a,b,c\}$, and $q=xc|ab$, $xa|bc$, $xb|ac$ with equal probabilities.
Suppose, $a$ and $b$ are always on one side of $e$, we consider the following four cases,
$$\begin{cases}
\textrm{case I}&c\in B \textrm{ and } \{a,b\} \subseteq A;\\
\textrm{case II}&\{a,b,c \} \subseteq B;\\
\textrm{case III}&c\in A \textrm{ and }  \{a,b\} \subseteq B;\\
\textrm{case IV}&\{a,b,c \} \subseteq A.\\
\end{cases}$$

We use $Q_{\rm I}$, $Q_{\rm II}$, $Q_{\rm III}$ and $Q_{\rm IV}$ to denote the set of quartet trees on leaf set $\{x,a,b,c\}$ in the case I, II, III and IV respectively, and let $Q$ be the entire set of quartet trees for the  leaf set of $\{x,a,b,c\}$. Since the four cases are mutually exclusive, $Q_i$s partition $Q$, $i\in\{\rm I,II,III,IV \}$, and the sizes of $Q_i$s are $|Q_{\rm I}|=\binom{k}{2}\times\binom{n-k}{1}$, $|Q_{\rm II}|=\binom{n-k}{3}$, $|Q_{\rm III}|=\binom{n-k}{2}\times\binom{k}{1}$ and $|Q_{\rm IV}|=\binom{k}{3}$.

Let $w(e)$ be a random variable of the weight that is added to $e$ for a quartet tree of $\{x,a,b,c\}$. Consider $w(e)$ for each case $\{\rm I,II,III,IV \}$. Then we have:
\begin{itemize}
\item case I and III: $w(e)=\begin{cases}
 1,&w.p.~\frac{2}{3};\\ 0&w.p.~\frac{1}{3},
 \end{cases}$
\item case II and IV: $w(e)=0$.
\end{itemize}

Let $W_{i}(e)$, $i\in\{\rm I,II,III,IV \}$, be the sum of all the weights added to the edge $e$. $W_{\rm I}(e)$ is a binomial random variable with parameters $\binom{k}{2}\binom{n-k}{1}$ and $\frac{2}{3}$; $W_{\rm III}(e)$ is a binomial random variable with parameters $\binom{n-k}{2}\binom{k}{1}$ and $\frac{2}{3}$; $W_{\rm II}=W_{\rm IV}=0$. Let $W_n(e)$ be the sum of $W_i(e)$ values, so we have $W_n(e)=W_{\rm I}(e)+W_{\rm III}(e)$. Let $n_1=\binom{k}{2}\binom{n-k}{1}$, and $n_2=\binom{n-k}{2}\binom{k}{1}$, then $$n_1+n_2=\frac{k(n-k)(n-2)}{2},$$ and so $W_n(e)$ consists of this many independent trials with probability of success on each trial of $\frac{2}{3}$.
That is, $W_n(e)$ is a binomial random variable with parameters $\frac{k(n-k)(n-2)}{2}$ and $\frac{2}{3}$.

\end{proof}

\subsection*{\bf Proof of inequality (1)}

Let $E_n^{\rm P}$ denote the set of pendent edges of current $X_n$-tree $T_n$, and $E_n^{\rm I}$ be the set of interior edges.

\begin{lemma}\label{lemma:E_diff}
For any $e''\in E_n^{\rm P}$ and any $e'\in E_n^{\rm I}$, the expected pendant edge total weight $W_n(e'')$ and the expected interior edge total weight $W_n(e')$, satisfy the inequality:
\begin{equation}
\mathbb{E}\left[W_n(e')\right]-\mathbb{E}\left[W_n(e'')\right]\geq \frac{1}{3}\left[n^2-5n+6\right]>0.
\label{ineq:1}
\end{equation}
\end{lemma}
\begin{proof}

$W_n(e'')$ and $W_n(e')$ are binomial random variables with the same probability of success $\frac{2}{3}$, but different number of trials $\binom{n-1}{2}$ and $\frac{k(n-k)(n-2)}{2}$, where $k \in \{2, \ldots, n-2\}$.  Thus
$$\mathbb{E}[W_n(e'')]=\frac{2}{3}\binom{n-1}{2},\qquad \mathbb{E}[W_n(e')]=\frac{2}{3}\frac{k(n-k)(n-2)}{2}.$$ 


For a fixed $n$, $\mathbb{E}\left[W_n(e')\right]-\mathbb{E}\left[W_n(e'')\right]$ is a function of $k$. Therefore, to find the minimum of the difference between these two expected values, we need to find the value(s) of $k$ for which $\mathbb{E}\left[W_n(e')\right]-\mathbb{E}\left[W_n(e'')\right]$ is minimal.

Let $y=(n-2)(n-k)k-(n^2-3n+2)$, then $\frac{dy}{dk}=(n-2)(n-2k)$.
When $k=\frac{n}{2}$, $\frac{dy}{dk}=0$, $\frac{d^2y}{dk^2}<0$. Thus, there is a maximum at $k=\frac{n}{2}$, 
and minimum occurs at $k=2$ or $k=n-2$.
Therefore, when $k=2$ or $k=n-2$,
$$
\frac{1}{3}\left[n^2-5n+6\right]\leq \mathbb{E}\left[W_n(e')\right]-\mathbb{E}\left[W_n(e'')\right]
$$
Moreover, it is easily shown that for $n>3$, $\frac{1}{3}\left[n^2-5n+6\right]>0$.
Therefore, $$\mathbb{E}\left[W_n(e')\right]-\mathbb{E}\left[W_n(e'')\right]\geq \frac{1}{3}\left[n^2-5n+6\right]>0.$$

\end{proof}

\begin{theorem}
\label{thm:2}
For any $e''\in E_n^{\rm P}$ and any $e'\in E_n^{\rm I}$,
\begin{equation*}
\mathbb{P}\left(W_n(e'')\geq W_n(e')\right)\leq 2\exp(-\frac{1}{576}n).
\end{equation*}
\end{theorem}
\begin{proof}
Let $W_n''=W_n(e'')-\mathbb{E}\left[W_n(e'')\right]$, \\$W_n'=W_n(e')-\mathbb{E}\left[W_n(e')\right]$, and $\beta=\mathbb{E}\left[W_n(e')\right]-\mathbb{E}\left[W_n(e'')\right]$. By {\bf Lemma \ref{lemma:E_diff}}, for $n\geq 4$, $\beta\geq 2dn^2$, where $d=\frac{1}{48}$. \\

Now, 
\begin{align*}
\mathbb{P}\left(W_n(e'')\geq W_n(e')\right) = & \mathbb{P}\left(W_n''- W_n' \geq \beta \right),\\
\leq &  \mathbb{P}\left(W_n''  \geq \frac{\beta}{2}\text{ or } -W_n'  \geq \frac{\beta}{2}\right), \\
\leq & \mathbb{P}\left(W_n''  \geq \frac{\beta}{2}\right)+\mathbb{P}\left( -W_n'  \geq \frac{\beta}{2}\right),\\
\leq & \mathbb{P}\left(W_n''  \geq dn^2 \right)+\mathbb{P}\left( -W_n'  \geq dn^2 \right).\\
\end{align*}

We now apply Hoeffding's Inequality to the two terms on the right. Suppose that $\{ Y_i, i =  1, 2, 3, ...,N \}$ are independent Bernoulli random variables, and let $Y=\sum_{i=1}^N Y_i$. By  Hoeffding's Inequality (Hoeffding 1963), we have:
$$\mathbb{P}\left(Y-\mathbb{E}(Y) \geq t\right) \leq \exp\left (-2t^2/N  \right),$$
$$\mathbb{P}\left(-(Y-\mathbb{E}(Y)) \geq t\right) \leq \exp\left (-2t^2/N \right). $$
Taking $Y=W_n'$ (and $W_n''$), $t=dn^2$ , and $N=\frac{k(n-k)(n-2)}{2}$ in the previous string of inequalities, gives:
\begin{align*}
 \mathbb{P}\left(W_n(e'')\geq W_n(e')\right) &\leq 2\exp(-\frac{1}{576\frac{k}{n}(1-\frac{k}{n})(1-\frac{2}{n})}n) ,\\ &\leq 2\exp(-\frac{1}{576}n).
\end{align*}

\end{proof}

\subsection*{\bf Proof of Inequality (2)}

\begin{proof}
We will use Theorem \ref{thm:2} to establish  Inequality (2). 
For $e''\in E_n^{\rm P}$, and $e'\in E_n^{\rm I}$, let $D$ be the event that $\displaystyle\min_{e''\in E_n^{\rm P}}\{W_n(e'')\}<\min_{e'\in E_n^{\rm I}}\{W_n(e')\}$,

Consider the complement of the event $D$, $$D^c=\left(\min_{e\in E_n^{\rm P}}\{W_n(e'')\}<\min_{e'\in E_n^{\rm I}}\{W_n(e')\}\right)^c,$$
that is there is an interior edge $e'$, such that \\$\displaystyle W_n(e')<\min_{e''\in E_n^{\rm P}}\{W_n(e'')\}$,  $W_n(e')\leq W_n(e'')$, $\forall e''\in E_n^{\rm P}$.
Let $A_{e'',e'}$ be the event that $W_n(e'')>W_n(e')$,
then we have, $\displaystyle D^c\subseteq \bigcup_{(e'',e')\in P\times I}A_{e'',e'}$, and so
$$\mathbb{P}\left(D^c\right)\leq \mathbb{P}\left(\bigcup_{(e'',e')\in P\times I}A_{e'',e'}\right).$$

According to Boole's inequality, 
\begin{equation}\mathbb{P}\left(\bigcup_{(e'',e')\in P\times I}A_{e'',e'}\right)\leq\sum_{(e'',e')\in P\times I}\mathbb{P}(A_{e'',e'}).
\label{ieq:boole_ineq}
\end{equation}
Now,  the number of pendent edge is $n$, i.e. $|P|=n$, and the number of interior edge is $n-3$, i.e. $|I|=n-3$. Thus, $|P\times I|=n(n-3)$, and so, by Theorem \ref{thm:2}, $\mathbb{P}(A_{e'',e'})=\mathbb{P}\left(W_n(e'')\geq W_n(e')\right)\leq 2\exp(-\frac{1}{576}n)$.
Thus,
\begin{equation}
\sum_{(e'',e')\in P\times I}\mathbb{P}(A_{e'',e'})\leq n(n-3)2\exp(-\frac{1}{576}n)\leq 2n^2 \exp(-\frac{1}{576}n).
\label{ieq:boole_ineq2}
\end{equation}
Therefore,
$$\mathbb{P}\left(\min_{e''\in E_n^{\rm P}}\{W_n(e'')\} \leq \min_{e'\in E_n^{\rm I}}\{W_n(e')\}\right)\geq 1- 2 n^2\exp(-\frac{1}{576}n).$$

\end{proof}

\subsection*{\bf Proof of Inequality (3)}

\begin{proof}
Since $\frac{k^2\exp(-ck)}{\exp(-ck/2)}=k^2\exp(-ck/2)$, and  $k^2\exp(-ck/2)\leq 1$ for $c\geq\frac{4\log{k}}{k}$ and $k> 1$, we have:
$$k^2\exp(-ck)\leq \exp(-\frac{c}{2} k), \quad \text{where } c\geq\frac{4\log{k}}{k} \textrm{ and } k> 1.$$
Thus
$\sum^\infty_{k=m}k^2 \exp(-ck)\leq \sum^\infty_{k=m} \exp(-\frac{c}{2}k), \quad \text{where } c\geq\frac{4\log{k}}{k} \textrm{ and } k> 1.$
where $\displaystyle\sum^\infty_{k=m}\exp(-\frac{c}{2}k)$ is the sum of a geometric series,
$$\sum^\infty_{k=m}\exp(-\frac{c}{2}k)= \frac{\exp(-cm/2)}{1-\exp(-c/2)}.$$
For $m\geq m_0$, $\exp(-cm/2)\leq \exp(-cm_0/2)$.
Therefore, \\ $\displaystyle\sum^\infty_{k=m}k^2 \exp(-ck)\leq \frac{\exp(-cm_0/2)}{1-\exp(-c/2)}$, where $c\geq\frac{4\log{k}}{k}$ and $k> 1$.
\end{proof}


\begin{thebibliography}{}

\bibitem{Bagchi1985}
Bagchi, A. and A.~K. Pal (1985).
\newblock Asymptotic normality in the generalized polya-eggenberger urn model,
  with an application to computer data structures.
\newblock {\em SIAM Journal on Algebraic and Discrete Methods\/}~{\em 6\/}(3).


\bibitem{Daubin2004}
Daubin, V. and H.~Ochman (2004).
\newblock Quartet mapping and the extent of lateral transfer in bacterial
  genomes.
\newblock {\em Molecular Biology and Evolution\/}~{\em 1}, 86--89.


\bibitem{Dress86}
Dress, A., A.~{von Haeseler}, and M.~Krueger (1986).
\newblock Reconstructing phylogenetic trees using variants of the ``four-point
  condition''.
\newblock {\em Studien zur Klassifikation\/}~(17), 299--305.


\bibitem{Felsenstein1981}
Felsenstein, J. (1981).
\newblock Evolutionary trees from dna sequences: A maximum likelihood approach.
\newblock {\em Journal of Molecular Evolution\/}~{\em 17}, 368--376.


\bibitem{Guindon2010}
Guindon, S., J.-F. Dufayard, V.~Lefort, M.~Anisimova, W.~Hordijk, and
  O.~Gascuel (2010).
\newblock New algorithms and methods to estimate maximum-likelihood
  phylogenies: Assessing the performance of phyml 3.0.
\newblock {\em Systematic Biology\/}~{\em 59\/}(3), 307--321.


\bibitem{Guindon2003}
Guindon, S. and O.~Gascuel (2003).
\newblock A simple, fast, and accurate algorithm to estimate large phylogenies
  by maximum likelihood.
\newblock {\em Systematic Biology\/}~{\em 52\/}(5), 696--704.


\bibitem{Jordan1869}
Jordan, C. (1869).
\newblock Sur les assemblages des lignes.
\newblock {\em Journal f\"{u}r die reine und angewandte Mathematik\/}~{\em 70},
  185--190.


\bibitem{Kang1975}
Kang, A. N.~C. and D.~A. Ault (1975).
\newblock Some properties of a centroid of a free tree.
\newblock {\em Information Processing Letters\/}~{\em 4\/}(1), 18--20.


\bibitem{Karr1993}
Karr, A.~F. (1993).
\newblock {\em Probability}.
\newblock Springer-Verlag.


\bibitem{Kotz2000}
Kotz, S., N.~Balakrishnan, N.~L, and Johnson (2000).
\newblock {\em Continuous Multivariate Distributions, Volume 1, Models and
  Applications\/} (2 ed.).
\newblock New York: Wiley.


\bibitem{Mahmoud2008}
Mahmoud, H. (2008).
\newblock {\em P\'{o}lya Urn Models}.
\newblock Chapman \& Hall / CRC.


\bibitem{McKenzie2000}
McKenzie, A. and M.~Steel (2000).
\newblock Distributions of cherries for two models of trees.
\newblock {\em Mathematical Biosciences\/}~{\em 164}, 81--92.


\bibitem{Mitchell1978}
Mitchell, S.~L. (1978).
\newblock Another characterization of the centroid of a tree.
\newblock {\em Discrete Mathematics\/}~{\em 24}, 277--280.


\bibitem{Nieselt2001}
Nieselt-Struwe, K. and A.~{von Haeseler} (2001).
\newblock Quartet-mapping, a generalization of the likelihood-mapping
  procedure.
\newblock {\em Molecular Biology and Evolution\/}~{\em 7\/}(18), 1204--1219.


\bibitem{Schmidt2002}
Schmidt, H.~A., K.~Strimmer, M.~Vingron, and A.~{von Haeseler} (2002).
\newblock {TREE-PUZZLE}: Maximum likelihood phylogenetic analysis using
  quartets and parallel computing.
\newblock {\em Bioinformatics\/}~{\em 18\/}(3), 502--504.


\bibitem{Smythe1996}
Smythe, R.~T. (1996).
\newblock Central limit theorems for urn models.
\newblock {\em Stochastic Processes and their Applications\/}~{\em 65},
  115--137.


\bibitem{Steel1993parsimony}
Steel, M.~A. (1993).
\newblock Distributions on bicoloured binary trees arising from the principle
  of parsimony.
\newblock {\em Discrete Applied Mathematics\/}~{\em 43}, 245--261.


\bibitem{Strimmer1997}
Strimmer, K., N.~Goldman, and A.~{von Haeseler} (1997).
\newblock Bayesian probabilities and quartet puzzling.
\newblock {\em Molecular Biology and Evolution\/}~{\em 2\/}(14), 210--211.


\bibitem{Strimmer1996}
Strimmer, K. and A.~{von Haeseler} (1996).
\newblock Quartet puzzling: a quartet maximum-likelihood method for
  reconstructing tree topologies.
\newblock {\em Molecular Biology and Evolution\/}.


\bibitem{Vinh2011}
Vinh, L.~S., A.~Fuehrer, and A.~{von Haeseler} (2011).
\newblock {Random Tree-Puzzle} leads to the {Yule-Harding Distribution}.
\newblock {\em Molecular Biology and Evolution\/}~{\em 28\/}(2), 873--877.


\end{thebibliography}
\end{document}